\documentclass[pra,twocolumn,showpacs,10pt,floatfix]{revtex4-1}

\usepackage{amsmath}
\usepackage{amsfonts}
\usepackage{amssymb}
\usepackage{mathrsfs}
\usepackage{latexsym}
\usepackage{dsfont}
\usepackage{graphicx}

\newenvironment{theorem}[1][Theorem]{\vspace{0.3cm}\noindent\textbf{#1.} }{\vspace{0.2cm}}
\newenvironment{definition}[1][Definition]{\vspace{0.3cm}\noindent\textbf{#1.} }{\vspace{0.2cm}}
\newenvironment{proof}[1][Proof]{\noindent\textbf{#1.} }{\ \rule{0.5em}{0.5em}}

\begin{document}

\title{Global Quantum Discord in Multipartite Systems } 

\author{C. C. \surname{Rulli}}
\email{clodoaldorulli@gmail.com}
\author{M. S. \surname{Sarandy}}
\email{msarandy@if.uff.br}

\affiliation{Instituto de F\'{\i}sica, Universidade Federal Fluminense, Av. Gal. Milton Tavares de Souza s/n, Gragoat\'a, 
24210-346, Niter\'oi, RJ, Brazil.}

\date{\today }

\begin{abstract}
We propose a global measure for quantum correlations in multipartite systems, 
which is obtained by suitably recasting the quantum discord in terms of relative 
entropy and local von Neumann measurements. The measure is symmetric with respect 
to subsystem exchange and is shown to be non-negative for an arbitrary state. 
As an illustration, we consider tripartite correlations in the Werner-GHZ state and 
multipartite correlations at quantum criticality. In particular, in contrast with the 
pairwise quantum discord, we show that the global quantum discord is able to characterize 
the infinite-order quantum phase transition in the Ashkin-Teller spin chain. 

\end{abstract}


\maketitle

%%%%%%%%%%%%%%%%%%%%%%%%%
\section{Introduction}
%%%%%%%%%%%%%%%%%%%%%%%%%

Quantum correlations constitute a fundamental resource for quantum information 
tasks~\cite{Nielsen:00}. They are rooted in the superposition principle, 
displaying effects with no classical analog. The research on quantum correlation measures 
was initially developed based on the entanglement-separability paradigm~\cite{Amico:08}. 
More recently, however, it has been perceived that entangled states are not the 
only kind of quantum states exhibiting nonclassical features. 
In this context, a suitable measure of quantum correlation has been introduced by 
Ollivier and Zurek~\cite{Ollivier:01}. This measure, which has been designated as 
{\it quantum discord}, is able to capture not only quantum correlations in 
entangled states but also in separable states. It arises as a difference between 
two expressions for the total correlation in a bipartite system (as measured by the 
mutual information), which are classically equivalent but distinct in the quantum 
regime. Remarkably, quantum discord has been revealed as a useful quantity in a number 
of applications, such as quantum critical phenomena~\cite{QPTApp1,QPTApp2} and  
quantum evolution under decoherence~\cite{DecApp}. Moreover, quantum discord has also been 
conjectured to be a resource for speed up in  quantum computation~\cite{DQC1} 
and for locking classical correlations in quantum states~\cite{locking}.

In recent years, generalizations of quantum discord to multipartite states have been 
considered in different scenarios~\cite{Modi:11}. One possible approach is based on 
directly generalizing the quantum mutual information to a multipartite 
system, even though nonunique generalizations are possible in this situation~\cite{Chakrabarty:10,Okrasa:11}. 
Another approach is to define from the beginning a measure based on the relative entropy, which allows for a 
unified view of different correlation sources, such as entanglement, quantum discord, and dissonance~\cite{Modi:10} 
(see also Ref.~\cite{Giogi:11}). 
The aim of this work is to propose a global measure of quantum discord obtained by a systematic extension of the 
bipartite quantum discord, with operational appeal and satisfying the basic requirements of a correlation function. 
In this direction, we suitably recast the standard bipartite quantum discord defined in Ref.~\cite{Ollivier:01} in 
terms of relative entropy and local von Neumann measurements, whence a natural multipartite measure 
for quantum correlations emerges. This measure -- named here as {\it global quantum discord} (GQD) -- is symmetric 
with respect to subsystem exchange and shown to be non-negative for arbitrary states. 
We illustrate our results by computing the tripartite GQD in the Werner-GHZ state and by applying GQD in the 
characterization of the infinite-order quantum phase transition (QPT) in the Ashkin-Teller spin 
chain, where ordinary pairwise quantum discord fails. 

%%%%%%%%%%%%%%%%%%%%%%%%%
\section{Quantum Discord}
%%%%%%%%%%%%%%%%%%%%%%%%%

Consider a bipartite system $AB$ composed of subsystems 
$A$ and $B$. Denoting by $\hat{\rho}_{AB}$ the density operator of $AB$ and by $\hat{\rho}_A$ 
and $\hat{\rho}_B$ the density operator of parts $A$ and $B$, 
respectively, the total correlation between $A$ and $B$ is measured by the quantum mutual information 
\begin{equation}
I(\hat{\rho}_{AB}) = S(\hat{\rho}_A) - S(\hat{\rho}_A | \hat{\rho}_B),
\end{equation}
where $S(\hat{\rho}_A) = -{\textrm{Tr}} \hat{\rho}_A \log_2 \hat{\rho}_A$ is the von Neumann 
entropy for $A$ and 
\begin{equation}
S(\hat{\rho}_A | \hat{\rho}_B) = S(\hat{\rho}_{AB}) - S(\hat{\rho}_B)
\end{equation}
is the entropy of $A$ conditional on $B$. The conditional entropy can also be introduced by a 
measurement-based approach. Indeed, consider a measurement locally performed on $B$, which 
can be described by a set of projectors $\{\hat{\Pi}_{B}^{j}\}=\{|b_j\rangle\langle b_j|\}$. 
The state of the quantum system, conditioned on the measurement of the outcome labeled 
by $j$, is
\begin{equation}
\hat{\rho}_{AB|j} = \frac{1}{p_j} \left( \hat{\mathbf{1}}_A \otimes \hat{\Pi}_{B}^{j} \right) \hat{\rho}_{AB} 
\left( \hat{\mathbf{1}}_A \otimes \hat{\Pi}_{B}^{j} \right), 
\end{equation}
where $p_j = {\textrm{Tr}} [ (\hat{\mathbf{1}}_A \otimes \hat{\Pi}_{B}^{j} ) \hat{\rho}_{AB}(\hat{\mathbf{1}}_A \otimes \hat{\Pi}_{B}^{j} )]$ 
denotes the probability of obtaining the outcome $j$ and $\hat{\mathbf{1}}_A$ denotes the identity operator for $A$. 
The conditional density operator $\hat{\rho}_{AB|j}$ allows for the following 
alternative definition of the conditional entropy:
\begin{equation}
S(\hat{\rho}_{AB}|\{\hat{\Pi}_{B}^{j}\}) = \sum_j p_j S(\hat{\rho}_{A|j}),
\end{equation}
where $\hat{\rho}_{A|j} = {\textrm{Tr}}_B\, \hat{\rho}_{AB|j} = (1/p_j)\langle b_j| \hat{\rho}_{AB} |b_j\rangle$, 
with $S(\hat{\rho}_{A|j})=S(\hat{\rho}_{AB|j})$.
Therefore, the quantum mutual information can also be defined by
\begin{equation}
J(\hat{\rho}_{AB}) = S(\hat{\rho}_A) - S(\hat{\rho}_{AB}|\{\hat{\Pi}_{B}^{j}\}). 
\end{equation}
The quantities $I(\hat{\rho}_{AB})$ and $J(\hat{\rho}_{AB})$ are classically equivalent 
but they are distinct in the quantum case. This difference is the quantum discord 
$\overline{\mathcal{D}}\left( \hat{\rho}_{AB}\right)$~\cite{Ollivier:01}, yielding
\begin{equation}
\overline{\mathcal{D}}\left( \hat{\rho}_{AB}\right) =I\left( \hat{\rho}_{AB}\right) -J\left( \hat{\rho}_{AB}\right).
\end{equation}
Note that $\overline{\mathcal{D}}\left( \hat{\rho}_{AB}\right)$ is defined as a non-negative asymmetric  
quantity that depends on ${\{\hat{\Pi}_{B}^{j}\}}$. This dependence can be eliminated  
by minimizing $\overline{\mathcal{D}}\left( \hat{\rho}_{AB}\right)$ over all measurement 
bases $\{\hat{\Pi}_{B}^{j}\}$~\cite{Henderson:01}. 

%%%%%%%%%%%%%%%%%%%%%%%%%%%%%%%%%%%%%%%%%%%%%%%%%%%%%%%%%%%%%
\section{Relative entropy and symmetric quantum discord}
%%%%%%%%%%%%%%%%%%%%%%%%%%%%%%%%%%%%%%%%%%%%%%%%%%%%%%%%%%%%%
The quantum relative entropy is a measure of distinguishability 
between two arbitrary density operators $\hat{\rho}$ and $\hat{\sigma}$, 
which is defined as~\cite{Vedral:02}
\begin{equation}
S\left( \hat{\rho}\parallel \hat{\sigma}\right) =\mathrm{Tr}\left( \hat{\rho}%
\log_2 \hat{\rho}-\hat{\rho}\log_2 \hat{\sigma}\right)\, .
\end{equation}
We can express the quantum mutual information $I(\hat{\rho}_{AB})$ 
as the relative entropy between $\hat{\rho}_{AB}$ and the product state 
$\hat{\rho}_{A}\otimes\hat{\rho}_{B}$, i.e.
\begin{equation}
I\left( \hat{\rho}_{AB}\right) =S\left( \hat{\rho}_{AB}\parallel \hat{\rho}_{A}\otimes 
\hat{\rho}_{B}\right).
\label{I-relative}
\end{equation}
In order to express the measurement-induced quantum mutual information 
$J\left( \hat{\rho}_{AB}\right)$ in terms of relative entropy, we need to consider 
a non-selective von Neumann measurement on part $B$ of $\hat{\rho}_{AB}$, 
which yields
\begin{eqnarray}
\Phi _{B}\left( \hat{\rho}_{AB}\right) &=& \sum_{j}
\left( \hat{1}_{A}\otimes \hat{\Pi}_{B}^{j} \right) 
\hat{\rho}_{AB}
\left(\hat{1}_{A}\otimes \hat{\Pi}_{B}^{j}\right) \nonumber \\
&=& \sum_{j}p_{j}\hat{\rho}_{A|j}\otimes \left\vert b_{j}\right\rangle
\left\langle b_{j}\right\vert.
\end{eqnarray}
Moreover, tracing over the variables of the subsystem $A$, we obtain
\begin{equation}
\Phi _{B}\left( \hat{\rho}_{B}\right) =
\Phi_{B}\left( \mathrm{Tr}_{A}\,\hat{\rho}_{AB}\right) 
=\sum_{j}p_{j}\left\vert b_{j}\right\rangle \left\langle b_{j}\right\vert,
\end{equation}
where we have used that $\mathrm{Tr}_{A} (\hat{\rho}_{A|j})=1$. Then, 
by expressing the entropies 
$S\left( \Phi _{B}\left( \hat{\rho}_{AB}\right)\right) $ and 
$S\left( \Phi _{B}\left( \hat{\rho}_{B}\right) \right) $ as
\begin{equation}
S\left( \Phi _{B}\left( \hat{\rho}_{AB}\right) \right) = H\left( \mathbf{p}%
\right) +\sum_{j}p_{j}S\left( \hat{\rho}_{A|j}\right)
\end{equation}
and 
\begin{equation}
S\left( \Phi _{B}\left( \hat{\rho}_{B}\right) \right) =H\left( \mathbf{p}%
\right), 
\end{equation}
with $H\left( \mathbf{p}\right)$ denoting the Shannon entropy
\begin{equation}
H\left( \mathbf{p}\right) =-\sum_{j}p_{j}\log _{2}\left( p_{j}\right),
\end{equation}
we can rewrite $J(\hat{\rho}_{AB})$ as
\begin{eqnarray}
J\left( \hat{\rho}_{AB}\right) &=& S\left( \hat{\rho}_{A}\right) -\sum_{j}p_{j}S\left( 
\hat{\rho}_{A|j}\right) \nonumber \\
\hspace{-0.6cm}&=&S\left( \hat{\rho}_{A}\right) +S\left( \Phi _{B}\left( \hat{\rho}%
_{B}\right) \right) -S\left( \Phi _{B}\left( \hat{\rho}_{AB}\right) \right)
\nonumber \\
\hspace{-0.6cm}&=&S\left( \Phi _{B}\left( \hat{\rho}_{AB}\right) \parallel \hat{\rho}%
_{A}\otimes \Phi _{B}\left( \hat{\rho}_{B}\right) \right). 
\end{eqnarray}
Therefore, the quantum discord can be rewriten in terms of a 
difference of relative entropies: 
\begin{eqnarray}
\overline{\mathcal{D}}\left( \hat{\rho}_{AB}\right) &=& S\left( \hat{\rho}_{AB}\parallel \hat{\rho}%
_{A}\otimes \hat{\rho}_{B}\right) \nonumber \\
&&-S\left( \Phi _{B}\left( \hat{\rho}%
_{AB}\right) \parallel \hat{\rho}_{A}\otimes \Phi _{B}\left( \hat{\rho}%
_{B}\right) \right),
\end{eqnarray}
with minimization taken over $\{\hat{\Pi}_{B}^{j}\}$ to remove 
the measurement-basis dependence.
It is possible then to obtain a natural symmetric extension $\mathcal{D}\left(\hat{\rho}_{AB}\right)$ 
for the quantum discord $\overline{\mathcal{D}}\left( \hat{\rho}_{AB}\right)$. 
Indeed, performing measurements over both subsystems 
$A$ and $B$, we define
\begin{eqnarray}
\mathcal{D}\left(\hat{\rho}_{AB}\right) &=& \min_{\{\hat{\Pi}_{A}^{j}\otimes\hat{\Pi}_{B}^{k}\}} 
\left[  S\left( \hat{\rho}_{AB}\parallel \hat{\rho}%
_{A}\otimes \hat{\rho}_{B}\right) \right. \nonumber \\
&&\left.\hspace{-0.6cm}-S\left( \Phi _{AB}\left( \hat{\rho}%
_{AB}\right) \parallel \Phi _{A}\left( \hat{\rho}_{A}\right) \otimes \Phi
_{B}\left( \hat{\rho}_{B}\right) \right) \right] \, ,  \label{DiscordiaBipartite}
\end{eqnarray}%
where the operator $\Phi _{AB}$ is given by
\begin{equation}
\Phi _{AB}\left( \hat{\rho}_{AB}\right) =\sum_{j,k} \left(\hat{\Pi}_{A}^{j}\otimes 
\hat{\Pi}_{B}^{k} \right) \hat{\rho}_{AB} \left(\hat{\Pi}_{A}^{j}\otimes \hat{\Pi}_{B}^{k}\right) \, .
\end{equation}
Observe that, by writing Eq.~(\ref{DiscordiaBipartite}) in terms of the mutual 
information $I$, we obtain 
\begin{equation}
\mathcal{D}\left( \hat{\rho}_{AB}\right) = 
\min_{\{\hat{\Pi}_{A}^{j}\otimes\hat{\Pi}_{B}^{k}\}} \left[
I(\hat{\rho}_{AB}) - 
I(\Phi _{AB}\left( \hat{\rho}_{AB}\right))\right], 
\end{equation}
which is the symmetric version of the 
expression for the loss of correlation due to measurement~\cite{Luo:10,Okrasa:11}. 
Remarkably, $\mathcal{D}\left( \hat{\rho}_{AB}\right)$ is equivalent to the 
measurement-induced disturbance (MID)~\cite{Luo:08} if measurement is performed in the 
eigenprojectors of the reduced density operators of each part (instead of minimization).
Moreover, Eq.~(\ref{DiscordiaBipartite}) also provides the symmetric 
quantum discord considered in Ref.~\cite{Maziero:10} and experimentally witnessed 
in Ref.~\cite{Auccaise:11}. 
As a further step, we can still rearrange Eq.~(\ref{DiscordiaBipartite}) 
in a rather convenient way, yielding
\begin{eqnarray}
\mathcal{D}\left( \hat{\rho}_{AB}\right) &=& \min_{\{\hat{\Pi}_{A}^{j}\otimes\hat{\Pi}_{B}^{k}\}} 
\left[ S\left( \hat{\rho}_{AB}\parallel \Phi
_{AB}\left( \hat{\rho}_{AB}\right) \right) \right. \nonumber \\
&&\left.\hspace{-0.6cm}-S\left( \hat{\rho}_{A}\parallel
\Phi _{A}\left( \hat{\rho}_{A}\right) \right) -S\left( \hat{\rho}%
_{B}\parallel \Phi _{B}\left( \hat{\rho}_{B}\right) \right)\right] .  \label{DD}
\end{eqnarray}

%%%%%%%%%%%%%%%%%%%%%%%%%%%%%%%%%%
\section{Global quantum discord}
%%%%%%%%%%%%%%%%%%%%%%%%%%%%%%%%%%

Let us now extend quantum discord as given by Eq.~(\ref{DD}) to multipartite systems. 

\begin{definition}
{\it The global quantum discord $\mathcal{D}\left( \hat{\rho}_{A_1 \cdots A_N} \right)$ for 
an arbitrary multipartite state $\hat{\rho}_{A_1 \cdots A_N}$ under a set of local 
measurements $\{\hat{\Pi}_{A_1}^{j_1} \otimes \cdots \otimes \hat{\Pi}_{A_N}^{j_N}\}$ is defined as}
\begin{eqnarray}
\mathcal{D}\left( \hat{\rho}_{A_1 \cdots A_N} \right) &=& \min_{\{\hat{\Pi}_{k}\}} [\,
S\left( \hat{\rho}_{A_1 \cdots A_N} \parallel \Phi\left( \hat{\rho}_{A_1 \cdots A_N} \right) \right) 
\nonumber \\
&&- \sum_{j=1}^{N}S\left( \hat{\rho}_{A_j} \parallel \Phi _{j} \left( \hat{\rho}_{A_j}\right) \right) \,], 
\label{gqd-def}
\end{eqnarray}
{\it where
$\Phi_j\left( \hat{\rho}_{A_j} \right) = \sum_{j^\prime} \hat{\Pi}_{A_j}^{j^\prime} \, \hat{\rho}_{A_j} \,
\hat{\Pi}_{A_j}^{j^\prime}$ and 
$\Phi\left( \hat{\rho}_{A_1 \cdots A_N} \right) = \sum_{k} {\hat{\Pi}}_{k} \, \hat{\rho}_{A_1 \cdots A_N} \, 
{\hat{\Pi}}_{k}$,
with $\hat{\Pi}_{k} = \hat{\Pi}_{A_1}^{j_1} \otimes \cdots \otimes \hat{\Pi}_{A_N}^{j_N}$ and 
$k$ denoting the index string $(j_1 \cdots j_N$)}. 
\end{definition}

Therefore, a classical state can be defined by $\hat{\rho}_{A_1 \cdots A_N} = \Phi\left( \hat{\rho}_{A_1 \cdots A_N} \right)$, 
which is in agreement with the requirement that classical states are not disturbed by suitable local measurements. 
Indeed, this definition of a classical state implies that $\hat{\rho}_{A_j}=\Phi _{j} \left( \hat{\rho}_{A_j}\right)$ 
for any $j$, which means $\mathcal{D}\left( \hat{\rho}_{A_1 \cdots A_N} \right)=0$. 
Moreover, observe that, via minimization over the set of projectors $\{\hat{\Pi}_{A_1}^{j_1} \otimes \cdots \otimes \hat{\Pi}_{A_N}^{j_N}\}$,  
we define GQD as a measurement-basis independent quantity. However, as will be illustrated in the Ashkin-Teller chain, other (non-minimizing) 
bases are also able to provide relevant information about the behavior of quantum correlations in the system (similarly to the 
original definition of quantum discord in Ref.~\cite{Ollivier:01}).
In any case, we can show that GQD is non-negative for an arbitrary state. 

\begin{theorem}
{\it The global quantum discord $\mathcal{D}\left( \hat{\rho}_{A_1 \cdots A_N} \right)$ is 
non-negative, i.e., $\mathcal{D}\left( \hat{\rho}_{A_1 \cdots A_N}  \right) \geqslant 0$}.
\end{theorem}

\begin{proof}
In order to prove that $\mathcal{D}\left( \hat{\rho}_{A_1 \cdots A_N}  \right) \geqslant 0$, 
we associate with each subsystem $A_j$ an ancilla system $B_j$. Therefore, we will define  
a composite density operator $\hat{\rho}_{A_1 \cdots A_N;B_1\cdots B_N}^\prime$ such that 
\begin{equation}
\hat{\rho}_{A_1 \cdots A_N;B_1\cdots B_N}^\prime = \sum_k \sum_{k^\prime} 
{\hat{\Pi}}_{k} \, \hat{\rho}_{A_1 \cdots A_N} \, {\hat{\Pi}}_{k^\prime} \otimes 
\hat{\Lambda}_{k k^\prime},
\label{rho-ancilla}
\end{equation}
where $\hat{\Lambda}_{k k^\prime} = |B_{j_1}\rangle \langle B_{j^\prime_1}| \otimes \cdots \otimes 
|B_{j_N}\rangle \langle B_{j^\prime_N}|$, 
with $k$ and $k^\prime$ denoting the index strings $(j_1 \cdots j_N$) and $(j^\prime_1 \cdots j^\prime_N$), 
respectively. From the monotonicity of the relative entropy under partial trace~\cite{Ruskai:02}, 
for any positive operators $\hat{\sigma}_{12}$ and $\hat{\gamma}_{12}$ such that 
$\mathrm{Tr}\left( \hat{\sigma}_{12}\right) =\mathrm{Tr}\left( \hat{\gamma}_{12}\right)$, 
we have that $S\left( \hat{\sigma}_{12}\Vert \hat{\gamma}_{12}\right) \geqslant S\left( 
\hat{\sigma}_{1}\Vert \hat{\gamma}_{1}\right)$, where 
$\hat{\sigma}_{1} =\mathrm{Tr}_{2}\left( \hat{\sigma}_{12}\right)$ and 
$\hat{\gamma}_{1} =\mathrm{Tr}_{2}\left( \hat{\gamma}_{12}\right)$. Then
$S\left( \hat{\sigma}_{123\ldots N}\Vert \hat{\gamma}_{123\ldots N}\right)
\geqslant \ldots \geqslant S\left( \hat{\sigma}_{123}\Vert \hat{\gamma}_{123}\right) 
\geqslant S\left( \hat{\sigma}_{12}\Vert \hat{\gamma}_{12}\right) \geqslant S\left( \hat{\sigma}_{1}\Vert \hat{\gamma}_{1}\right)$.
By taking $\hat{\rho}_{A_1 \cdots A_N;B_1\cdots B_N}^\prime$ as $\hat{\sigma}$ and 
$\hat{\rho}_{A_1;B_1}\otimes \hat{\rho}_{A_2;B_2}^\prime \otimes \ldots \otimes \hat{\rho}_{A_N;B_N}^\prime$ 
as $\hat{\gamma}$, we obtain
\begin{eqnarray}
&&\hspace{-1.3cm}S\left( \hat{\rho}_{A_1 \cdots A_N;B_1\cdots B_N}^\prime \Vert 
\hat{\rho}_{A_1;B_1}^\prime \otimes \hat{\rho}_{A_2;B_2}^\prime \otimes \ldots \otimes \hat{\rho}_{A_N;B_N}^\prime \right) \nonumber \\
\hspace{1cm} &\geqslant& S\left( \hat{\rho}_{A_1 \cdots A_N}^\prime \Vert \hat{\rho}_{A_1}^\prime \otimes \ldots \otimes \hat{\rho}_{A_N}^\prime \right), 
\label{ineq-1}
\end{eqnarray}
which, from Eq.~(\ref{I-relative}), implies that
\begin{eqnarray}
\sum_{j=1}^N S\left( \hat{\rho}_{A_j;B_j}^\prime \right) - 
S\left( \hat{\rho}_{A_1 \cdots A_N;B_1\cdots B_N}^\prime \right) 
\nonumber \\
\geqslant \sum_{j=1}^N S\left( \hat{\rho}_{A_j}^\prime \right) - 
S\left( \hat{\rho}_{A_1 \cdots A_N}^\prime \right).
\label{Nentropias}
\end{eqnarray}
However, from Eq.~(\ref{rho-ancilla}), it follows the relations
\begin{eqnarray}
&&S\left( \hat{\rho}_{A_1 \cdots A_N;B_1\cdots B_N}^\prime \right) = S\left( \hat{\rho}_{A_1 \cdots A_N}\right) \, , \label{prime1}\\
&&S\left( \hat{\rho}_{A_1 \cdots A_N}^{\prime }\right) = S\left( \Phi \left( \hat{\rho}_{A_1 \cdots A_N}\right) \right) \, \label{prime2},  \\ 
&&S\left( \hat{\rho}_{A_j;B_j}^\prime \right) = S\left( \hat{\rho}_{A_j} \right) \,\,\,\,\,\,\,\,\,\, (\forall j) \, , \label{prime3}\\
&&S\left( \hat{\rho}_{A_j}^{\prime }\right) = S\left( \Phi _{j}\left( \hat{\rho}_{A_j}\right) \right) \,\,\,\, (\forall j) \label{prime4}.
\end{eqnarray}
Insertion of Eqs.(\ref{prime1})-(\ref{prime4}) into inequality~(\ref{Nentropias}) yields
\begin{eqnarray}
&&\hspace{-0.3cm}\sum_{j=1}^{N} S\left( \hat{\rho}_{A_j}\right) - S\left( \hat{\rho}_{A_1 \cdots A_N} \right) \nonumber \\ 
&&\geqslant \sum_{j+1}^{N} S\left( \Phi_{j}\left( \hat{\rho}_{A_j }\right) \right) 
- S\left( \Phi \left( \hat{\rho}_{A_1 \cdots A_N}\right)  \right) \, .
\label{intermed1}
\end{eqnarray}
By rewriting inequality (\ref{intermed1}) in terms of the relative entropy, we obtain
\begin{equation}
S\left( \hat{\rho}_{A_1 \cdots A_N} \Vert \Phi \left( \hat{\rho}_{A_1 \cdots A_N}\right) \right) -
\sum_{j=1}^{N} S\left( \hat{\rho}_{A_j} \Vert \Phi _{j}\left( \hat{\rho}_{A_j}\right) \right) \geqslant 0. 
\end{equation}
The left hand side of the inequality above is exactly the GQD, as defined by Eq.(\ref{gqd-def}). Hence, 
$\mathcal{D}\left( \hat{\rho}_{A_1 \cdots A_N} \right) \geqslant 0$.
\end{proof}

%%%%%%%%%%%%%%%%%%%%%%%%%%%%%%%%%%%%%%%%%%%%%%%%%%%%%%%%%%%
\section{Tripartite correlations in the Werner-GHZ state}
%%%%%%%%%%%%%%%%%%%%%%%%%%%%%%%%%%%%%%%%%%%%%%%%%%%%%%%%%%%

As a first illustration of GQD, we will consider the Werner-GHZ state
\begin{equation}
\hat{\rho} = \frac{\left( 1-\mu \right) }{8}\hat{\mathbf{1}}+\mu \left\vert
GHZ\right\rangle \left\langle GHZ\right\vert , \label{werner-ghz}
\end{equation}
where $0 \le \mu \le 1$ and 
\begin{equation}
\left\vert GHZ\right\rangle =(\left\vert \uparrow \right\rangle
_{A}\left\vert \uparrow \right\rangle _{B}\left\vert \uparrow \right\rangle
_{C}+\left\vert \downarrow \right\rangle _{A}\left\vert \downarrow
\right\rangle _{B}\left\vert \downarrow \right\rangle _{C})/\sqrt{2},
\label{ghz-state}
\end{equation} 
with $\left\vert \uparrow \right\rangle$ and $\left\vert \downarrow \right\rangle$ 
denoting the eigenstates of the Pauli operator $\hat{\sigma}^z$ associated with eigenvalues 
$1$ and $-1$, respectively. 
The Werner-GHZ state provides an interpolation between a fully mixed (uncorrelated) 
state and a maximally correlated pure tripartite state. It is a rather suitable  
state to begin with as we propose a measure for quantum correlation and constitutes 
an interesting scenario to compare multipartite with bipartite correlations, since 
it is a generalization of the two-qubit Werner state~\cite{Werner:89}. 
Let us begin by analyzing GQD in the case of a pure GHZ state ($\mu=1$).

%%%%%%%%%%%%%%%%%%%%%%%%%%%%%%%%%%%%%%%%%%%%%%%
\subsection{GQD for the GHZ state}
%%%%%%%%%%%%%%%%%%%%%%%%%%%%%%%%%%%%%%%%%%%%%%%

Let us focus here on the $GHZ$ state, as defined by Eq.~(\ref{ghz-state}). 
In order to define local measurements for $|GHZ\rangle$, let us consider 
rotations in the directions of the basis vectors of subsystems $A$, $B$, and $C$, 
which are denoted by
\begin{eqnarray}
|+\rangle _{j} &=&\cos \left( \frac{\theta _{j}}{2}\right)
| \uparrow \rangle _{j}+{\textrm{e}}^{i\varphi_j}\sin \left( \frac{\theta _{j}}{2}%
\right) | \downarrow \rangle _{j} \, , \\
| -\rangle _{j} &=&-{\textrm{e}}^{-i\varphi_j}\sin \left( \frac{\theta _{j}}{2}\right)
| \uparrow \rangle _{j}+\cos \left( \frac{\theta _{j}}{2}%
\right) | \downarrow \rangle _{j} \, ,
\end{eqnarray}%
with $j=1,2,3$ for subsystems $A$, $B$, and $C$, respectively. The angles $\theta_i$ take 
values in the interval $[0,\pi)$ and the angles $\varphi_i$ take values in the interval $[0,2\pi)$. 
In order to compute $\mathcal{D}\left( \hat{\rho} \right)$, with 
$\hat{\rho} = \left\vert GHZ\right\rangle \left\langle GHZ\right\vert$, we must 
evaluate the expression
\begin{eqnarray}
\mathcal{D}\left( \hat{\rho} \right) &=& \min_{\theta_i,\varphi_i} \left[\,
S\left( \hat{\rho} \parallel \Phi\left( \hat{\rho} \right) \right) 
- S\left( \hat{\rho}_{A} \parallel \Phi _{A} \left( \hat{\rho}_{A}\right) \right)\right.
\nonumber \\
&&\left. - S\left( \hat{\rho}_{B} \parallel \Phi _{B} \left( \hat{\rho}_{B}\right) \right)
- S\left( \hat{\rho}_{C} \parallel \Phi _{C} \left( \hat{\rho}_{C}\right) \right)\right].
\label{gqd-tri}
\end{eqnarray}
However, $S\left( \hat{\rho}\right) =0$, since $\hat{\rho}$ is pure. 
Moreover, $S\left( \hat{\rho}_{A}\parallel \Phi _{A}\left( \hat{\rho}_{A}\right)
\right) =S\left( \hat{\rho}_{B}\parallel \Phi _{B}\left( \hat{\rho}%
_{B}\right) \right) =S\left( \hat{\rho}_{C}\parallel \Phi _{C}\left( \hat{%
\rho}_{C}\right) \right) =0$, since $\hat{\rho}_{A}$, $\hat{\rho}_{B}$, and 
$\hat{\rho}_{C}$ are proportional to identity operators. Hence, GQD is simply given by
\begin{equation}
\mathcal{D}\left( \hat{\rho} \right) = \min_{\theta_i,\varphi_i} 
S\left( \Phi\left( \hat{\rho}\right) \right) = \min_{\theta_i,\varphi_i}\left[
-\sum_{j}\lambda _{j}\log _{2}\lambda_{j}\right], \label{D-GHZ}
\end{equation}
where $\lambda_j$ are the eigenvalues of the operator $\Phi\left( \hat{\rho}\right)$. 
They can be obtained from projections of the GHZ state over the rotated basis states. 
In order to minimize $S\left( \Phi\left( \hat{\rho}\right) \right)$, we must find out the 
measurement basis that maximizes the purity of $\Phi\left( \hat{\rho}\right)$, i.e., that 
maximizes the dispersion of the eigenvalues $\lambda_j$ with respect to the average of $\{\lambda_j\}$. 
This is obtained for $\theta_i=0$ (i=1,2,3), namely, measurements in the eigenprojectors of $\sigma^z_i$. 
As an illustration, let us consider the case of $\theta_1=0$ and $\varphi_i=0$ (i=1,2,3). 
In this situation, the eigenvalues $\lambda_j$ for the operator $\Phi\left( \hat{\rho}\right)$ 
read
\begin{eqnarray}
\lambda _{1}&=&\lambda _{8}=\frac{1}{2}\cos ^{2}\left( \frac{\theta _{2}}{2}%
\right) \cos ^{2}\left( \frac{\theta _{3}}{2}\right), \label{lambda-1}\\
\lambda _{2}&=&\lambda _{7}=\frac{1}{2}\cos ^{2}\left( \frac{\theta _{2}}{2}%
\right) \sin ^{2}\left( \frac{\theta _{3}}{2}\right), \label{lambda-2} \\
\lambda _{3}&=&\lambda _{6}=\frac{1}{2}\sin ^{2}\left( \frac{\theta _{2}}{2}%
\right) \cos ^{2}\left( \frac{\theta _{3}}{2}\right), \label{lambda-3} \\
\lambda _{4}&=&\lambda _{5}=\frac{1}{2}\sin ^{2}\left( \frac{\theta _{2}}{2}%
\right) \sin ^{2}\left( \frac{\theta _{3}}{2}\right). \label{lambda-4}
\end{eqnarray}
By using Eqs.~(\ref{lambda-1})-(\ref{lambda-4}) into Eq.~(\ref{D-GHZ}), we can directly 
obtain $\mathcal{D}\left( \hat{\rho} \right)$ by minimizing over the angles 
$\theta _{2}$ and $\theta _{3}$. The function $D(\theta_2,\theta_3)$ to be minimized is then 
\begin{equation}
D(\theta_2,\theta_3) = -\sum_{j}\lambda _{j}\log _{2}\lambda_{j}.
\end{equation}
We plot $D(\theta_2,\theta_3)$ as a function of $\theta_2$ and $\theta_3$ in Fig.~\ref{f1}. 
Notice that its minimum, which provides $\mathcal{D}\left( \hat{\rho} \right)$, occurs at 
the boundary values $\theta _{2}=\theta _{3}=0$, where  
$\mathcal{D}\left( \hat{\rho} \right)=1$. This is a manisfestation of the fact 
that any local measurement disturbs the GHZ state, which is detected by a nonvanishing GQD. 
 
\begin{figure}[!ht]
 \centering
\includegraphics[scale=0.36]{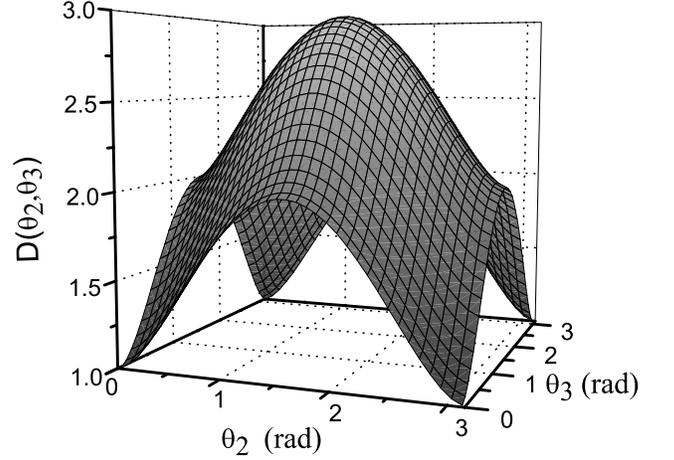}
\caption{The function $D(\theta_2,\theta_3)$ as a function of $\protect\theta _{2}$ 
and $\protect\theta _{3}$. Note that the minimum occurs for 
$\protect\theta _{2}=\protect\theta _{3}=0$, which implies 
$\mathcal{D}\left( \hat{\rho} \right) = 1$.}
\label{f1}
\end{figure}

%%%%%%%%%%%%%%%%%%%%%%%%%%%%%%%%%%%%%%%%%%%%%%%%%%%%%%%%%%%%%%
\subsection{GQD in the Werner-GHZ state}
%%%%%%%%%%%%%%%%%%%%%%%%%%%%%%%%%%%%%%%%%%%%%%%%%%%%%%%%%%%%%%

In order to obtain $\mathcal{D}\left( \hat{\rho}\right)$ for the Werner-GHZ 
state, let us first rewrite Eq.~(\ref{werner-ghz}) as 
\begin{eqnarray}
\hat{\rho} &=&\frac{1}{8}\hat{\mathbf{1}}+\frac{\mu }{8} \left( \hat{\sigma}_{1}^{z}\hat{\sigma}%
_{2}^{z}+\hat{\sigma}_{1}^{z}\hat{\sigma}_{3}^{z}+\hat{\sigma}_{2}^{z}\hat{%
\sigma}_{3}^{z}+\hat{\sigma}_{1}^{x}\hat{\sigma}_{2}^{x}\hat{\sigma}_{3}^{x} \right. \nonumber \\
&&\left.-\hat{\sigma}_{1}^{x}\hat{\sigma}_{2}^{y}\hat{\sigma}_{3}^{y}-\hat{\sigma}%
_{1}^{y}\hat{\sigma}_{2}^{x}\hat{\sigma}_{3}^{y}-\hat{\sigma}_{1}^{y}\hat{%
\sigma}_{2}^{y}\hat{\sigma}_{3}^{x}\right) , \label{werner-ghz-2}
\end{eqnarray}
with $\sigma^{x}_i$, $\sigma^{y}_i$, and $\sigma^{z}_i$ denoting the Pauli matrices 
for the qubit $i$. 
Again, we will have here that 
$S\left( \hat{\rho}_{i}\parallel \Phi _{i}\left( \hat{\rho}_{i}\right)
\right) = 0 \,\, (i=A,B,C)$, since $\hat{\rho}_{A}$, $\hat{\rho}_{B}$, and 
$\hat{\rho}_{C}$ are proportional to identity operators. Therefore, 
\begin{equation}
\mathcal{D}\left( \hat{\rho}\right) = \min_{\theta_i,\varphi_i}  S\left( \hat{\rho}\Vert \Phi \left( \hat{%
\rho}\right) \right) = \min_{\theta_i,\varphi_i} \, \left[S\left( \Phi \left( \hat{\rho}\right) \right) -S\left( \hat{\rho}\right)\right].
\label{d-w}
\end{equation}
The von Neumann entropy $S\left( \hat{\rho}\right)$ is given by 
\begin{eqnarray}
S\left( \hat{\rho}\right) &=& 3-\frac{7}{8}\left( 1-\mu \right) \log _{2}\left(
1-\mu \right) \nonumber \\
&&-\frac{1}{8}\left( 1+7\mu \right) \log _{2}\left( 1+7\mu
\right)\, . \label{s-rho-w}
\end{eqnarray}
As for the GHZ state, we take local measurements in the $\hat{\sigma}^z$ eigenbasis for each particle 
to minimize $S\left( \Phi \left( \hat{\rho}\right) \right)$. 
Such an eigenbasis provides the maximum loss of correlation among the parts of $\rho$, which therefore 
minimizes GQD. Then, from Eq.~(\ref{werner-ghz-2}), we obtain
\begin{equation}
\Phi \left( \hat{\rho}\right) =\left( \frac{1-\mu }{8}\right) \hat{1}+\frac{%
\mu }{8}\left( \hat{1}+\hat{\sigma}_{1}^{z}\hat{\sigma}_{2}^{z}+\hat{\sigma}%
_{1}^{z}\hat{\sigma}_{3}^{z}+\hat{\sigma}_{2}^{z}\hat{\sigma}_{3}^{z}\right), 
\end{equation}
which implies
\begin{eqnarray}
S\left( \Phi \left( \hat{\rho}\right) \right) &=& 3-\frac{3}{4}\left( 1-\mu
\right) \log _{2}\left( 1-\mu \right) \nonumber \\
&&-\frac{1}{4}\left( 1+3\mu \right) \log
_{2}\left( 1+3\mu \right) \, . \label{s-phi-rho-w}
\end{eqnarray}
Insertion of Eqs.~(\ref{s-rho-w}) and (\ref{s-phi-rho-w}) into Eq.~(\ref{d-w}) 
yields
\begin{eqnarray}
\hspace{-0.4cm}\mathcal{D}\left( \hat{\rho}\right) &=& 
-\frac{1}{4}\left( 1+3\mu \right) \log _{2}\left( 1+3\mu \right) \nonumber \\
&&\hspace{-1.2cm}+\frac{1}{8}\left( 1-\mu \right) \log _{2}\left( 1-\mu \right) +\frac{1}{8}\left(
1+7\mu \right) \log _{2}\left( 1+7\mu \right)\, .
\end{eqnarray}
In Fig.~\ref{f2} we plot $\mathcal{D}\left( \hat{\rho}\right)$ as a function of $\mu$. 
Observe that GQD vanishes only for $\mu=0$, where $\hat{\rho}$ is 
a completely mixed state. Moreover, GQD is a monotonic function of $\mu$, acquiring its maximal value 
$\mathcal{D}\left( \hat{\rho}\right)=1$ for $\mu=1$, where $\hat{\rho}$ is the GHZ state. 
This result resembles the behavior of the bipartite Werner state~\cite{Ollivier:01,Luo:08-2}.
\begin{figure}[!ht]
 \centering
\includegraphics[scale=0.30]{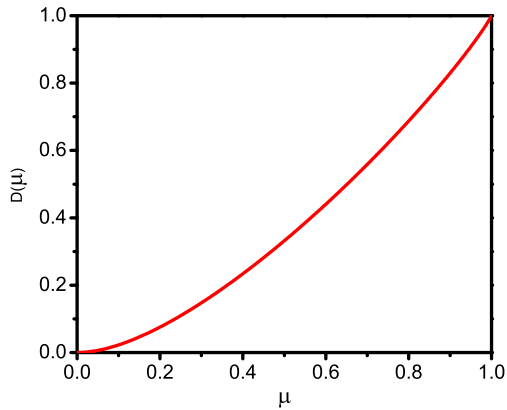}
\caption{(Color online) Tripartite GQD for the Werner-GHZ state as a function 
of the mixing parameter $\mu$. Note that GQD is nonvanishing for $\mu\ne 0$.}
\label{f2}
\end{figure}

%%%%%%%%%%%%%%%%%%%%%%%%%%%%%%%%%%%%%%%%%%%%%%%%%%%%%%%%%%%%%%%%%%%%%%%%%%%%%%%%%%
\section{Multipartite correlations in the Ashkin-Teller chain}
%%%%%%%%%%%%%%%%%%%%%%%%%%%%%%%%%%%%%%%%%%%%%%%%%%%%%%%%%%%%%%%%%%%%%%%%%%%%%%%%%%

Let us now present an application of GQD that makes evident the importance 
of considering genuine multipartite correlations to the characterization of a 
QPT. In this direction, we consider the 
Ashkin-Teller model, which has been introduced as a generalization of the Ising spin-1/2 
model to investigate the statistics of lattices with four-state 
interacting sites~\cite{Ashkin:43}. It exhibits a rich phase diagram~\cite{ATpd} 
and has recently attracted a great deal of attention  due to several interesting 
applications~\cite{ATapp}.
The Hamiltonian for the quantum Ashkin-Teller model in one-dimension for a chain with $M$ 
sites is given by 
\begin{eqnarray}
&&H_{AT} =-J\sum_{j=1}^{M}\left( \hat{\sigma}_{j}^{x}+\hat{\tau}_{j}^{x} + \Delta
\hat{\sigma}_{j}^{x} \hat{\tau}_{j}^{x}\right)  \nonumber \\
&&\hspace{-0.5cm}-J\,\beta \sum_{j=1}^{M}\left( \hat{\sigma}_{j}^{z}\hat{\sigma}_{j+1}^{z}
+ \hat{\tau_{j}}^{z}\hat{\tau}_{j+1}^{z} + \Delta \hat{\sigma_{j}}^{z}\hat{\sigma}_{j+1}^{z}
\hat{\tau}_{j}^{z}\hat{\tau}_{j+1}^{z}\right),  \label{at}
\end{eqnarray}
where $\hat{\sigma}_j^\alpha$ and $\hat{\tau}_j^\alpha$ $(\alpha = x,y,z)$ are
independent Pauli spin-1/2 operators, 
$J$ is the exchange coupling constant, 
$\Delta$ and $\beta$ are (dimensionless) parameters, 
and periodic boundary conditions (PBC) are adopted, i.e., $%
\hat{\sigma}^\alpha_{M+1} = \hat{\sigma}^\alpha_{1}$ and $\hat{\tau}^\alpha_{M+1} =
\hat{\tau}^\alpha_{1}$ ($\alpha = x,y,z$). 
The Ashkin-Teller model is $Z_2 \otimes
Z_2$ symmetric, with the Hamiltonian commuting with the parity operators 
\begin{equation}
\mathcal{P}_1 = \otimes_{j=1}^{M} \sigma_j^x \hspace{1cm} {\text{and}} \hspace{%
1cm} \mathcal{P}_2 = \otimes_{j=1}^{M} \tau_j^x .  \label{parity-at}
\end{equation}
Therefore, the eigenspace of $H_{AT}$ can be decomposed into four disjoint
sectors labeled by the eigenvalues of $\mathcal{P}_1$ and $\mathcal{P}_2$,
namely, $Q=0$ $(\mathcal{P}_1 = + 1, \mathcal{P}_2 = + 1)$, $Q=1$ $(\mathcal{%
P}_1 = + 1, \mathcal{P}_2 = - 1)$, $Q=2$ $(\mathcal{P}_1 = -1, \mathcal{P}_2
= - 1)$, and $Q=3$ $(\mathcal{P}_1 = -1, \mathcal{P}_2 = + 1)$. By the
symmetry of $H_{AT}$ under the interchange $\sigma^\alpha \leftrightarrow
\tau^\alpha$, the sectors $Q=1$ and $Q=3$ are degenerate. Moreover, we
observe that the ground state belongs to the sector $Q=0$. 
A schematic view of the Ashkin-Teller chain is shown in Fig.~\ref{f3}. Note that 
each site contains {\it two} spin particles, which means that the number $N$ of 
particles in a chain with $M$ sites is $N=2M$.
\begin{figure}[!ht]
 \centering
\includegraphics[scale=0.36]{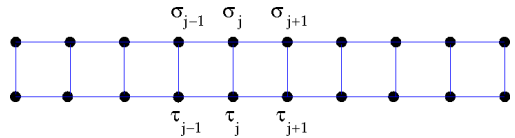}
\caption{(Color online) 
Schematic view of the Ashkin-Teller chain. The lattice is composed by two 
independent spin-1/2 particles \textit{per} site $j$ described by Pauli 
operators $\{\protect\sigma^ \protect\alpha_j, \protect\tau^ \protect\alpha_j\}$.}
\label{f3}
\end{figure}

The model presents an infinite-order quantum critical point at $\beta = 1$ and $\Delta = 1$. 
Infinite-order QPTs are typically detected by an extremum (either a maximum or a minimum) 
in quantum correlations measures (see, e.g., Refs.~\cite{Gu:03,Chen:06} for pairwise entanglement and 
Ref.~\cite{QPTApp1} for pairwise quantum discord). 
However, as shown in Ref.~\cite{Rulli:10}, pairwise entanglement is unable to characterize 
the critical point $(\beta,\Delta) = (1,1)$ in the Ashkin-Teller chain. Moreover, it can  
be shown that pairwise quantum discord does not detect such a QPT either. 
Indeed, taking $\beta=1$, the density operator $\hat{\rho}^{j,j}(\Delta)$ for a 
pair $\hat{\sigma}_j - \hat{\tau}_j$ is diagonal~\cite{Rulli:10} and exhibits vanishing 
quantum discord for any $\Delta$. For pairs $\hat{\sigma}_j - \hat{\sigma}_{j+1}$ (or $\hat{\tau}_j - \hat{\tau}_{j+1}$), 
the density operator $\hat{\rho}^{j,j+1}(\Delta)$ has off-diagonal terms. Such a state displays 
nonvanishing quantum discord. However, no identification (such as an extremum or a cusp) occurs 
at the critical point $\Delta=1$ (for any local measurement). 
\vspace{-0.5cm}
\begin{figure}[!ht]
 \centering
\includegraphics[scale=0.3]{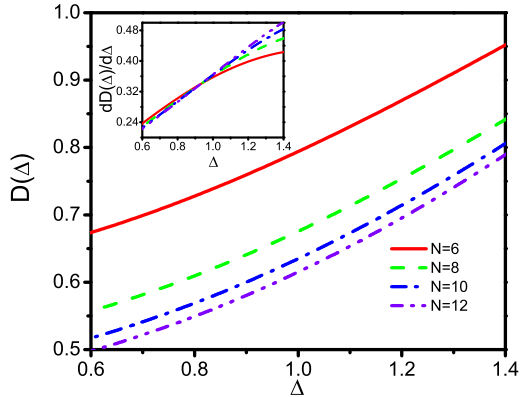}
\vspace{-0.5cm}
\caption{(Color online) 
GQD associated with the $\hat{\sigma}^z$ eigenbasis for a spin quartet in the Ashkin-Teller model for chains up to $N=12$ spins. Inset: Derivative 
of GQD with respect to $\Delta$.}
\label{f4}
\end{figure}

On the other hand, if we consider multipartite correlations, 
GQD is able to identify the QPT as an extremum at the critical point. 
However, such identification does not occur in the basis that minimizes GQD, which is given by the measurement of all spins in the 
$\hat{\sigma}^z$ eigenbasis. Instead, the infinite-order QPT turns out to be correctly characterized if, and only if, local measurements 
are performed in the $\hat{\sigma}^x$ eigenbasis. Remarkably, this is exactly the basis of eigenstates of the single spin 
reduced density operators. Therefore, computation of GQD in such an eigenbasis can be seen as a generalization of MID 
to the multipartite scenario.

\begin{figure}[!ht]
 \centering
\includegraphics[scale=0.3]{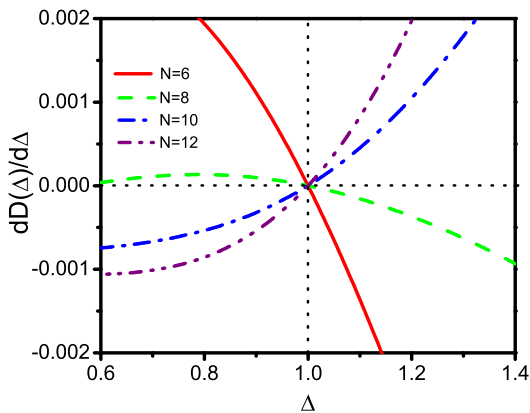}
\vspace{-0.5cm}
\caption{(Color online) 
Derivative of GQD associated with the $\hat{\sigma}^x$ eigenbasis for a spin quartet in the Ashkin-Teller model for chains up to $N=12$ spins.}
\label{f5}
\end{figure}

\begin{figure}[!ht]
 \centering
\includegraphics[scale=0.3]{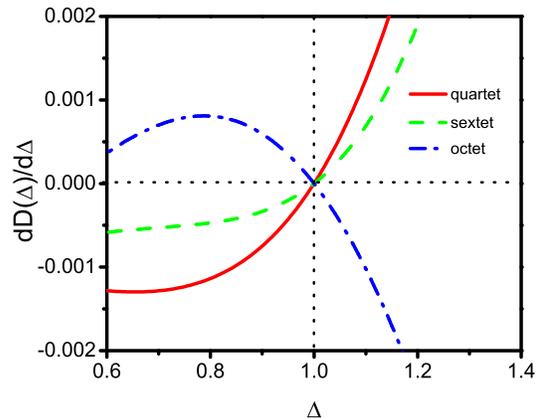}
\vspace{-0.5cm}
\caption{(Color online) 
Derivative of GQD associated with the $\hat{\sigma}^x$ eigenbasis for spin quartets, sextets, and octets in the Ashkin-Teller model for a chain with $N=16$ spins.}
\label{f6}
\end{figure}

We consider groups of $4$ particles 
(quartets) composed by spins $\hat{\sigma}_j - \hat{\sigma}_{j+1} - \hat{\tau}_j - \hat{\tau}_{j+1}$ as well as 
extensions for sextets 
and octets.
For those configurations, we numerically compute GQD relative to local measurements in the $\hat{\sigma}^z$ eigenbasis and 
in the $\hat{\sigma}^x$ eigenbasis for each particle via exact diagonalization of chains up to $16$ spin particles. 
The results are exhibited in Fig.~\ref{f4} for measurements in the $\hat{\sigma}^z$ eigenbasis and in Figs.~\ref{f5} and~\ref{f6} 
for measurements in the $\hat{\sigma}^x$ eigenbasis. We can observe that, in the $\hat{\sigma}^x$ eigenbasis, the identification of the QPT as 
an extremum (vanishing derivative of the GQD) already occurs for quartets in lattices with $N=6$ spins (see Fig.~\ref{f5}). 
Moreover, this characterization is kept for sextets and octets in larger chains (see Fig.~\ref{f6}). 

%%%%%%%%%%%%%%%%%%%%%%%%%%
\section{Conclusion}
%%%%%%%%%%%%%%%%%%%%%%%%%%

In summary, we have proposed a measure for multipartite quantum correlations. 
This measure has been obtained by suitably recasting the standard quantum discord in 
terms of relative entropy and local von Neumann measurements [as given by Eq.~(\ref{DD})].  
In particular, our measure is a systematic extension of the original approach for quantum discord as introduced in 
Ref.~\cite{Ollivier:01}, reducing to it in the particular case of bipartite systems.
Illustrations of its use have been provided for both the Werner-GHZ and the Ashkin-Teller spin chain. 
Further applications of GQD, such as the investigation of multipartite correlations in quantum computation 
and connections with entanglement (see, e.g., Ref.~\cite{Fanchini:10}), are left for future research.

%%%%%%%%%%%%%%%%%%%%%%%%%%%%%%
\begin{acknowledgments}
%%%%%%%%%%%%%%%%%%%%%%%%%%%%%%
We  acknowledge financial support from the Brazilian agencies CNPq and FAPERJ. 
This work was performed as part of the Brazilian National Institute for Science and 
Technology of Quantum Information (INCT-IQ).

\end{acknowledgments}

%%%%%%%%%%%%%%%%%%%%%%%%%%%%%


\begin{thebibliography}{99}  
%%%%%%%%%%%%%%%%%%%%%%%%%%%%%                                                                                             

\bibitem{Nielsen:00}  
M. A. Nielsen and I. L. Chuang, {\it Quantum Computation and
Quantum Information}, Cambridge University Press, 2000.

\bibitem {Amico:08} 
L. Amico {\it et al.}, Rev. Mod. Phys. \textbf{80}, 517 (2008). 

\bibitem {Ollivier:01}
H. Ollivier and W. H. Zurek, Phys. Rev. Lett. \textbf{88}, 017901 (2001).

\bibitem{QPTApp1} 
R. Dillenschneider, Phys, Rev. B \textbf{78}, 224413 (2008);
M. S. Sarandy, Phys. Rev. A \textbf{80}, 022108 (2009); 
T. Werlang {\it et al.}, Phys. Rev. Lett. \textbf{105}, 095702 (2010).

\bibitem {QPTApp2}
J. Maziero {\it et al.}, Phys. Rev. A \textbf{82}, 012106 (2010);
Y.-X. Chen and S.-W. Li, Phys. Rev. A \textbf{81}, 032120 (2010); 
B. Tomasello {\it et al.}, e-print arXiv:1012.4270 (2010).

\bibitem {DecApp}
A. Shabani and D. A. Lidar, Phys. Rev. Lett. \textbf{102}, 100402 (2009);
J. Maziero {\it et al.}, Phys. Rev. A \textbf{80}, 044102 (2009); 
J. Maziero {\it et al.}, Phys. Rev. A \textbf{81}, 022116 (2010); 
A. Ferraro {\it et al.}, Phys. Rev. A \textbf{81}, 052318 (2010); 
L. Mazzola, J. Piilo, and S. Maniscalco , Phys. Rev. Lett. \textbf{104}, 200401 (2010).

\bibitem {DQC1} 
A. Datta, A. Shaji, and C. M. Caves, Phys. Rev. Lett. \textbf{100}, 050502 (2008). 

\bibitem{locking} D. P. DiVincenzo {\it et al.}, Phys. Rev. Lett. {\bf 92}, 067902 (2004); 
A. Datta and S. Gharibian, Phys. Rev. A {\bf 79}, 042325 (2009); 
S. Boixo {\it et al.}, e-print arXiv:1105.2768 (2011).

\bibitem{Modi:11} K. Modi and V. Vedral, e-print arXiv:1104.1520 (2011).

\bibitem{Chakrabarty:10} I. Chakrabarty, P. Agrawal, and A. K. Pati, e-print arXiv:1006.5784 (2010).

\bibitem{Okrasa:11} M. Okrasa and Z. Walczak, e-print arXiv:1101.6057 (2011).

\bibitem{Modi:10} K. Modi {\it et al.}, Phys. Rev. Lett. {\bf 104}, 080501 (2010).

\bibitem{Giogi:11} G. L. Giorgi {\it et al.}, e-print ArXiv:1108.0434v1 (2011).

\bibitem{Henderson:01} L. Henderson and V. Vedral, J. Phys. A {\bf 34}, 6899 (2001).

\bibitem{Vedral:02} V. Vedral, Rev. Mod. Phys. {\bf 74}, 197 (2002).

\bibitem{Luo:10} S. Luo and S. Fu, Phys. Rev. A {\bf 82}, 034302 (2010).

\bibitem{Luo:08} S. Luo, Phys. Rev. A {\bf 77}, 022301 (2008).


\bibitem{Maziero:10} 
J. Maziero, L. C. Celeri, and R. M. Serra, e-print arXiv:1004.2082 (2010).

\bibitem{Auccaise:11} 
R. Auccaise {\it et al.}, Phys. Rev. Lett. {\bf 107}, 070501 (2011).

\bibitem{Ruskai:02} M. B. Ruskai, J. Math. Phys. {\bf 43}, 4358 (2002).

\bibitem{Werner:89} R. F. Werner, Phys. Rev. A {\bf 40}, 4277 (1989).

\bibitem{Luo:08-2} S. Luo, Phys. Rev. A {\bf 77}, 042303 (2008).

\bibitem{Ashkin:43} J. Ashkin and E. Teller, Phys. Rev. \textbf{64}, 178
(1943).

\bibitem{ATpd} 
M. Kohmoto, M. den Nijs, and L. P. Kadanoff, Phys. Rev. B \textbf{24}, 5229 (1981).

\bibitem{ATapp} C. Gils, J. Stat. Mech. {\bf P07019} (2009); 
M. S. Gr\o nsleth {\it et al.}, Phys. Rev. B \textbf{79}, 094506 (2009). 

\bibitem{Gu:03} S.-J. Gu, G.-S. Tian, and H.-Q. Lin, Phys. Rev. A {\bf 68}, 042330 (2003).

\bibitem{Chen:06} Y. Chen {\it et al.}, New J. Phys. {\bf 8}, 97 (2006).

\bibitem{Rulli:10} C. C. Rulli and M. S. Sarandy, Phys. Rev. A {\bf 81}, 032334 (2010).

\bibitem{Fanchini:10} F. F. Fanchini {\it et al.}, Phys. Rev. A {\bf 84}, 012313 (2011).

\end{thebibliography}
\end{document}